\documentclass[runningheads]{llncs}
\usepackage{lineno}
\usepackage{underscore}
\usepackage[utf8]{inputenc}
\usepackage{amsmath}
\usepackage{bm}
\usepackage{semantic}
\usepackage{nicefrac}
\usepackage{marginnote}
\usepackage{amssymb}
\usepackage{tikz}
\usepackage{breakcites}

\title{Typed SLD-Resolution: Dynamic Typing for Logic Programming}
\titlerunning{Typed SLD Resolution}
\author{João Barbosa \and
Mário Florido \and
Vítor Santos Costa}
\authorrunning{J. Barbosa et. al.}
\institute{LIACC, INESC, Dep. de Ciência de Computadores \\
Faculdade de Ciências \\Universidade do Porto, Porto\\ \email{\{joao.barbosa,amflorido,vscosta\}@fc.up.pt}}
\begin{document}

\maketitle

\begin{abstract}
The semantic foundations for logic programming are usually separated into two different approaches. The operational semantics, which uses SLD-resolution, the proof method that computes answers in logic programming, and the declarative semantics, which sees logic programs as formulas and its semantics as models. Here, we define a new operational semantics called TSLD-resolution, which stands for Typed SLD-resolution, where we include a value “wrong”, that corresponds to the detection of a type error at run-time.  For this we define a new typed unification algorithm. Finally we prove the correctness of TSLD-resolution with respect to a typed declarative semantics.

\end{abstract}

\section{Introduction}

Types play an important role in the verification and debugging of
programming languages, and have been the subject of significant
research in the logic programming community \cite{DBLP:conf/iclp/Zobel87,DBLP:books/mit/pfenning92/DartZ92,DBLP:journals/corr/cs-LO-9810001,DBLP:conf/lics/FruhwirthSVY91,DBLP:books/mit/pfenning92/YardeniFS92,DBLP:journals/ai/MycroftO84,DBLP:conf/slp/LakshmanR91,DBLP:conf/lopstr/SchrijversBG08,Drabent02,DBLP:conf/iclp/SchrijversCWD08,DBLP:books/crc/chb/Hanus14,BarbosaFloridoCosta19,BarbosaFloridoCosta21}. Most research has been driven by the desire to perform compile-time
checking. One important line of this work views types as
approximation of the program semantics \cite{DBLP:conf/iclp/Zobel87,DBLP:books/mit/pfenning92/DartZ92,DBLP:books/mit/pfenning92/YardeniFS92,DBLP:conf/iclp/BruynoogheJ88,DBLP:conf/lics/FruhwirthSVY91}. A different approach relies on asking the user to provide the
type information, thus filtering the set of admissible programs \cite{DBLP:journals/ai/MycroftO84,DBLP:conf/slp/LakshmanR91,DBLP:conf/iclp/SchrijversCWD08,BarbosaFloridoCosta19}. In practice, static type-checking is not widely
used in actual Prolog systems, but Prolog systems do rely on dynamic
typing to ensure that system built-in parameters are called with
acceptable arguments, such as {\em is/2}. In fact, the Prolog ISO
standard defines a set of predefined types and typing
violations \cite{DBLP:books/daglib/0083128}.

Motivated by these observations, we propose a step forward in the 
dynamic type checking of logic programs: to extend unification with a type checking mechanism. Type checking will thus become a core part of the resolution engine. This extension enables the detection of several bugs in Prolog programs which are rather difficult to capture in the standard untyped language, such as the unintended switch of arguments in a predicate call. This approach can also be used to help in the tracing of bugs in the traditional Prolog {\em Four-Port Model} debugging approach. 
In classical programming languages, type-checking essentially captures the use of functions on arguments of a type different from the expected. The parallel in logic programming is to capture type errors in queries applied to arguments of a different type. This is the essence of our new operational mechanism, here called {\em Typed SLD (TSLD)}.
Following Milner's argument on wrong programs \cite{DBLP:journals/jcss/Milner78}, type errors will be denoted by an extra value, \emph{wrong}; unification may succeed, fail, or be \emph{wrong}. As discussed in prior work, a three-valued semantics provides a natural framework for describing the new unification \cite{BarbosaFloridoCosta19}. We name the new evaluation mechanism, that extends unification with type checking thus performing dynamic typing, {\em Typed SLD (TSLD)}.
Let us now present a simple example of the use of TSLD-resolution:
\begin{example}
Consider the following program consisting of the three facts: 
\begin{verbatim}
p(0).
p(1).
p(a).
\end{verbatim}
Let $\square$ stand for success, {\em false} for failure and {\em wrong} for a run-time error. Assuming that constants $1$ and $a$ have different types (in this case, $int$ and $atom$, respectively), the TSLD-tree for the query $p(1)$ is:
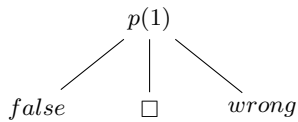
\begin{figure}[!h]
    \centering
\begin{tikzpicture}
\node {$p(1)$} [level distance = 1.2cm]
    child {node {$false$}} 
    child {node {$\square$}} 
    child {node {$wrong$}};
\end{tikzpicture}
\caption{TSLD-tree}
\end{figure}
\end{example}
In the following we prove that
TSLD-resolution is correct with respect to a new typed declarative semantics for logic programming based on the existence of several different semantic domains, for the
interpretation of terms, instead of the usual Herbrand universe domain.
The use of different domains for the semantics of logic programming is not
new \cite{DBLP:conf/slp/LakshmanR91}, but before it was used in a prescriptive typing approach where types were mandatory for every syntactic objects: functors, variables, and predicates.
Here we assume that only constants and function symbols have predefined types. 
Type correctness in the form of a reformulation of subject-reduction for SLD-resolution was defined in \cite{DBLP:conf/flops/DeransartS01} for a typed version of logic programs where also every syntactic objects must be typed statically. To semantically deal with dynamic typing here we define a new version of SLD (Typed SLD) and prove its soundness with respect to a new declarative semantics.
In another previous
work  \cite{DBLP:conf/iclp/SchrijversCWD08} dynamic type checking of Prolog predicates was done in two different scenarios: on calls from untyped to typed code using program transformation, and on calls to untyped code from typed code to check whether the untyped code satisfies previously made type annotations. In this previous work type annotations for predicates were necessary in both scenarios and the semantic soundness of these run-time checks was not studied. Here we do a semantic
study of dynamic typing and use it to show that a new operational mechanism
detecting run-time errors is sound. For this we use predefined types only for constants and function symbols.

The paper is organised as follows. Section 2 reviews some preliminary background concepts with the necessary definitions and results on the theory of types and logic programming, and sets the grounds for subsequent developments. In Section 3 we present the important new notion of {\em typed unification}, present the use of different {\em types} for constant symbols, and show that typed unification extends standard first order unification in the sense that when it succeeds the result is the same. Section 4 presents {\em TSLD-resolution}. We definite the notion of {\em TSLD-derivation step} and {\em TSLD-tree}, showing how run-time type errors are detected during program evaluation. We also distinguish the notion of a type error in the program from the notion of a type error in the query with respect to a program. In Section 5, we define a new declarative semantics for logic programming based on our previous work \cite{BarbosaFloridoCosta19} using a three-valued logic which uses the value {\em wrong} to denote run-time type errors, define ill-typed programs and queries, and show that TSLD-resolution is sound with respect to this declarative semantics. Finally, in section 6, we conclude and point out some research directions for future work.

\section{Preliminary Concepts}

In this section we will present concepts that are relevant both to the definitions of the operational and declarative semantics for logic programming. These concepts include the three-valued logic that will be used throughout the paper, defined initially in \cite{GlossarWiki:Kleene:1938}, and the definition of our syntax for types, terms, and programs.
%, and our syntax for types and type terms.

\subsection{Three-Valued Logic}

The three-valued logic used in this paper is the Weak Kleene logic \cite{GlossarWiki:Kleene:1938}, later interpreted by \cite{3value} and \cite{AJL}. As in our previous work \cite{BarbosaFloridoCosta19} the third value is called \emph{wrong} and represents a (dynamic) type error. In this logic, the value \emph{wrong} propagates through every connective, which is a behaviour we want the type error value to have. In table \ref{con}, we describe the connectives in the logic.

\begin{table}[!htb]
    \begin{minipage}{.50\linewidth}
        \begin{center}
        \begin{tabular}{| c | c | c | c |}
            \hline
            $\bm{\wedge}$ & \textit{true} & \textit{false} & \emph{wrong} \\ \hline
            \textit{true} & \textit{true} & \textit{false} & \emph{wrong} \\ \hline
            \textit{false} & \textit{false} & \textit{false} & \emph{wrong} \\ \hline
            \emph{wrong} & \emph{wrong} & \emph{wrong} & \emph{wrong}  \\ \hline
        \end{tabular}
        \end{center}
    \end{minipage}
    \begin{minipage}{.50\linewidth}
        \begin{center}
        \begin{tabular}{| c | c | c | c |}
            \hline
            $\bm{\vee}$ & \textit{true} & \textit{false} & \textit{wrong} \\ \hline
            \textit{true} & \textit{true} & \textit{true} & \textit{wrong} \\ \hline
            \textit{false} & \textit{true} & \textit{false} & \textit{wrong} \\ \hline
            \textit{wrong} & \textit{wrong} & \textit{wrong} & \textit{wrong}  \\ \hline
            \end{tabular}
            \end{center}
    \end{minipage}
    \\
\caption{Connectives of the three-valued logic - conjunction and disjunction}
\label{con}
\end{table}

The negation of logic values is defined as: $\neg true = false$, $\neg false = true$ and $\neg \emph{wrong} = \emph{wrong}$. And implication is defined as: $p --> q \equiv (\neg p) \vee q$.

Note that whenever the value \emph{wrong} occurs in any connective in the logic, the result of applying that connective is \emph{wrong}.

\subsection{Types}

In this paper we fix the set of base types {\em  int, float, atom} and {\em string}, an enumerable set of compound types $f(\sigma_1, \cdots, \sigma_n)$, where $f$ is a function symbol and $\sigma_i$ are types, and an enumerable set of functional types of the form $\sigma_1 \times \cdots \times \sigma_n \rightarrow \sigma$, where $\sigma_i$ and $\sigma$ are types.

We use this specific choice of base types because they correspond to types already present, to some extent, in Prolog. Some built-in predicates already expect integers, floating point numbers, or atoms.

\subsection{Terms}

The alphabet of logic programming is composed of symbols from disjoint classes. For our language of terms we have an infinite set of variables \textbf{Var}, an infinite set of function symbols \textbf{Fun}, parenthesis and the comma \cite{Apt:1996:LPP:249573}.

Terms are defined as follows:
\begin{itemize}
    \item a variable is a term,
    \item if $f$ is an n-ary function symbol and $t_1, \dots, t_n$ are terms, then $f(t_1,\dots,t_n)$ is a term,
    \item if $f$ is a function symbol of arity zero, then $f$ is a term and it is called a constant.
\end{itemize}

A ground term is a term with no variables.
In the rest of the paper we assume that ground terms are assumed to be typed, meaning that each constant has associated to it a base type and for any ground compound term $f(t_1, \cdots t_n)$, the function symbol $f$ (of arity $n \geq 1$) has associated to it a functional type of the form $\sigma_1 \times \cdots \times \sigma_n \rightarrow f(\sigma_1, \cdots, \sigma_n)$.
Note that variables are not statically typed, but, as we will see in the forthcoming sections, type checking of the use of variables will be made dynamically through TSLD-resolution.

\subsection{Programs and Queries}

We now extend our language of terms to a language of programs by adding an infinite set of predicate symbols \textbf{Pred} and the reverse implication $\leftarrow$.

The definition of atoms, queries, clauses and programs is the usual one \cite{Apt:1996:LPP:249573}:
\begin{itemize}
    \item an atom is either a predicate symbol $p$ with arity $n$, applied to terms $t_1,\dots,t_n$, which we write as $p(t_1,\dots,t_n)$% , or an equality of the form $t_1 = t_2$, where $t_1$ and $t_2$ are terms
    . We will represent atoms by $H,A,B$;
    \item a query is a finite sequence of atoms, which we will represent by $Q,\bar{A},\bar{B}$;
    \item a clause is of the form $H\leftarrow \bar{B}$, where $H$ is an atom of the form $p(t_1,\dots,t_n)$ and $\bar{B}$ is a query;
    \item a program is a finite set of clauses, which we will represent by $P$.
\end{itemize}

The interpretation of queries and clauses is quantified. Every variable that occurs in a query is assumed to be existentially quantified and every variable that occurs in a clause is universally quantified \cite{Apt:1996:LPP:249573}.

\section{Typed Unification}

Solving equality constraints using a unification algorithm \cite{Robinson1965,10.1145/357162.357169} is the main computational mechanism in logic programming. Logic programming usually uses an untyped term language and assumes a semantic universe composed of all semantic values: the Herbrand universe \cite{Herbrand1930,Apt:1996:LPP:249573}. 

However, in our work, we assume that the semantic values are split among several disjoint semantic domains and thus equality only makes sense inside each domain. Moreover each type will be mapped to a non-empty semantic domain. To reflect this, unification may now return three different outputs. Besides being successful or failing, unification can now return the \emph{wrong} value. This is the logical value of nonsense and reflects the fact that we are trying to perform unification between terms with different types corresponding to a type error during program evaluation.

A substitution is a mapping from variables to terms, which assigns to each variable $X$ in its domain a term $t$. We will represent bindings by $X \mapsto t$, substitutions by symbols such as $\theta, \eta, \delta \dots$, and applying a substitution $\theta$ to a term $t$ will be represented by $\theta(t)$. We say $\theta(t)$ is an instance of $t$.

Substitution composition is represented by $\circ$, i.e., the composition of the substitutions $\theta$ and $\eta$ is denoted $\theta \circ \eta$ and applying $(\theta \circ \eta)(t)$ corresponds to $\theta(\eta(t))$. We can also calculate substitution composition, i.e., $\delta = \theta \circ \eta$ as defined below \cite{Apt:1996:LPP:249573}.
\begin{definition}[Substitution Composition]
Suppose $\theta$ and $\eta$ are substitutions, such that $\theta = [X_1 \mapsto t_1, \dots, X_n \mapsto t_n]$ and $\eta = [Y_1 \mapsto t_1\prime, \dots, Y_m \mapsto t_m\prime]$. Then, composition $\eta \circ \theta$ is calculated by following these steps:
\begin{itemize}
    \item remove from the sequence $X_1 \mapsto \eta(t_1),\dots, X_n \mapsto \eta(t_n),Y_1 \mapsto t_1\prime,\dots,Y_m\mapsto t_m\prime$ the bindings $X_i \mapsto \eta(t_i)$ such that $X_i = \eta(t_i)$ and the elements $Y_i \mapsto t_i\prime$ for which $\exists X_j. Y_i = X_j$
    \item form a substitution from the resulting sequence.
\end{itemize}
\end{definition}

A substitution $\theta$ is called a unifier of two terms $t_1$ and $t_2$ iff $\theta(t_1) = \theta(t_2)$. If such a substitution exists, we say that the two terms are unifiable. In particular, a unifier $\theta$ is called a most general unifier (mgu) of two terms $t_1$ and $t_2$ if for every other unifier $\eta$ of $t_1$ and $t_2$, $\eta = \delta \circ \theta$, for some substitution $\delta$.

First order unification \cite{Robinson1965} assumes an untyped universe, so unification between any two terms always makes sense. Therefore, it either returns a mgu between the terms, if it exists, or halts with failure.

We argue that typed unification only makes sense between terms of the same type. Here we will extend a previous unification algorithm \cite{10.1145/357162.357169} to define a \emph{typed unification algorithm}, where failure will be separated into $false$, where two terms are not unifiable but may have the same type, and \emph{wrong}, where the terms cannot have the same type. 

\begin{definition}[Typed Unification Algorithm]\label{tus}
Let $t_1$ and $t_2$ be two terms, and $F$ be a flag that starts $true$. We create the starting set of equations as $S = \{t_1 = t_2\}$, and we will rewrite the pair $(S,F)$ by applying the following rules until it is no longer possible to apply any of them, or until the algorithm halts with \emph{wrong}. If no rules are applicable, then we output $false$ if the flag is $false$, or output the solved set $S$, which can be seen as a substitution.
\begin{enumerate}
    \item $(\{f(t_1,\dots,t_n) = f(s_1,\dots,s_n)\} \cup Rest,F) \to (\{t_1 = s_1,\dots,t_n=s_n\} \cup Rest,F)$
    \item $(\{f(t_1,\dots,t_n) = g(s_1,\dots,s_m)\} \cup Rest,F) \to wrong$, if $f \neq g$ or $n \neq m$
    \item $(\{c = c\} \cup Rest,F) \to (Rest,F)$
    \item $(\{c = d\}\cup Rest,F) \to (Rest,false)$, if $c\neq d$, and $c$ and $d$ have the same type
    \item $(\{c = d\}\cup Rest,F) \to wrong$, if $c \neq d$, and $c$ and $d$ have different types
    \item $(\{c = f(t_1,\dots,t_n)\}\cup Rest,F) \to wrong$
    \item $(\{f(t_1,\dots,t_n) = c\}\cup Rest,F) \to wrong$
    \item $(\{X = X\} \cup Rest,F) \to (Rest,F)$
    \item $(\{t = X\} \cup Rest,F) \to (\{X = t\} \cup Rest,F)$, where $t$ is not a variable and $X$ is a variable
    \item $(\{X = t\} \cup Rest,F) \to (\{X = t\} \cup [X \mapsto t](Rest),F)$, where $X$ does not occur in $t$ and $X$ occurs in $Rest$
    \item $(\{X = t\} \cup Rest,F) \to (Rest,false)$, where $X$ occurs in $t$ and $X \neq t$
\end{enumerate}

\end{definition}

Let us illustrate with an example.

\begin{example}
Let $t_1$ be $g(X,a,f(1))$ and $t_2$ be $g(b,Y,f(2))$. We generate the pair $(\{g(X,a,f(1))=g(b,Y,f(2))\},true)$, and proceed to apply the rewriting rules.\\
$(\{g(X,a,f(1))=g(b,Y,f(2))\},true) \to_1 \\(\{X = b, a = Y, f(1) = f(2)\},true) \to_{10} \\(\{X = b, Y = a, f(1) = f(2)\},true) \to_1 \\(\{X = b, Y = a, 1 = 2\},true) \to_5 \\(\{X = b, Y = a\},false) \to false$.
\end{example}
The successful cases of this algorithm are the same as for first order unification \cite{10.1145/357162.357169}. We will prove this result in the following theorem.

\begin{theorem}[Conservative with respect to term unification]\label{equi}
Let $t_1$ and $t_2$ be two terms. If we apply the Martelli-Montanari algorithm (MM algorithm) to $t_1$ and $t_2$ and it returns a set of solved equalities $S$, then the typed unification algorithm applied to the same two terms is also successful and returns the same set of equalities.
\end{theorem}
The following theorem proves typed unification detects run-time type errors.

\begin{theorem}[Ill-typed unification]\label{itu}
If the output of the typed unification algorithm is \emph{wrong}, then there is no substitution $\theta$ such that $\theta(t_1)$ and $\theta(t_2)$ have the same type.
\end{theorem}

\begin{example}
Let $t_1 = f(1,g(h(X,2)),Y)$ and $t_2 = f(Z,g(h(W,a)),1)$. The typed unification algorithm outputs $wrong$. We can see in Table \ref{trep} that there is no substitution $\theta$ such that $\theta(t_1) = \theta(t_2)$, nor any substitution $\theta$ such that $\theta(t_1)$ has the same type as $\theta(t_2)$, since the highlighted terms cannot have the same type for any substitution.
\begin{table}[!htb]
    \begin{minipage}{.50\linewidth}
        \begin{center}
        \begin{tikzpicture}
            [level 1/.style = {level distance = 1cm},
            level 2/.style = {level distance = 1cm}]
            \node {$f$}
                child {node {$1$}} 
                child {node {$g$}
                    child {node {$h$}
                        child {node {$X$}}
                        child {node [draw]{$2$}}
                    }
                } 
                child {node {$Y$}};
        \end{tikzpicture}
        \end{center}
    \end{minipage}
    \begin{minipage}{.50\linewidth}
        \begin{center}
        \begin{tikzpicture}
            [level 1/.style = {level distance = 1cm},
            level 2/.style = {level distance = 1cm}]
            \node {$f$}
                child {node {$Z$}} 
                child {node {$g$}
                    child {node {$h$}
                        child {node {$W$}}
                        child {node [draw]{$a$}}
                    }
                } 
                child {node {$1$}};        
        \end{tikzpicture}
        \end{center}
    \end{minipage}
    \\
\caption{Tree representation terms $t_1$ and $t_2$}
\label{trep}
\end{table}
\end{example}

\section{Operational Semantics}

The operational semantics of logic programming describes how answers are computed. Here we define Typed SDL-resolution (TSLD) which returns the third value {\em wrong} whenever it finds a type error. We start by defining a TSLD-derivation step, which is a variation on the basic mechanism for computing answers to queries in the untyped semantics for logic programming, the SLD-derivation step. The major difference is the use of the typed unification algorithm. Then we create TSLD-derivations by iteratively applying these singular steps. After this, we introduce the concept of TSLD-trees and use it to represent the search space for answers in logic programming. Finally, we interpret the contents of the TSLD-tree.

\subsection{TSLD-derivation}

To compute in logic programming, we need a program $P$ and a query $Q$. We can interpret $P$ as being a set of statements, or rules, and $Q$ as being a question that will be answered by finding an instance $\theta(Q)$ such that $\theta(Q)$ follows from $P$. The essence of computation in logic programming is then to find such $\theta$ \cite{Apt:1996:LPP:249573}.

In our setting the basic step for computation is the TSLD-derivation step. It corresponds to having a non-empty query $Q$ and selecting from $Q$ an atom $A$. If $A$ unifies with $H$, where $H\leftarrow \bar{B}$ is an input clause, we replace $A$ in $Q$ by $\bar{B}$ and apply an mgu of $A$ and $H$ to the query.
% If the selected atom $A$ is of the form $t_1 = t_2$, and $\theta$ is an mgu of $t_1$ and $t_2$, we delete $A$ from $Q$ and apply $\theta$ to the query.

\begin{definition}[TSLD-derivation step]
Consider a non-empty query $Q = \bar{A_1},A,\bar{A_2}$ and a clause $c$ of the form $H \leftarrow \bar{B}$. Suppose that $A$ unifies (using typed unification) with $H$ and let $\theta$ be a mgu of $A$ and $H$. $A$ is called the {\em selected atom} of $Q$. Then we write $$\bar{A_1},A,\bar{A_2} \underset{c}\implies \theta(\bar{A_1},\bar{B},\bar{A_2})$$ and call it a {\em TSLD-derivation step}. $H \leftarrow \bar{B}$ is called its {\em input clause}.
If typed unification between the selected atom $A$ and the input clause $c$ outputs $wrong$ (or $false$) we write the TSLD-derivation step as 
$Q \implies wrong$ (or $Q \implies false,\bar{A_1},\bar{A_2}$). 
\end{definition}

In this definition we assume that $A$ is variable disjoint with $H$. It is always possible to rename the variables in $H \leftarrow \bar{B}$ in order to achieve this, without loss of generality.

\begin{definition}[TSLD-derivation]
Given a program $P$ and a query $Q$ a sequence of TSLD-derivation steps from $Q$ with input clauses of $P$ reaching the empty query, $false$, or $wrong$, is called a {\em TSLD-derivation} of $Q$ in $P$.
\end{definition}
If the program is clear from the context, we speak of a TSLD-derivation of the query $Q$
and if the input clauses are irrelevant we drop the reference to them.
Informally, a TSLD-derivation corresponds to iterating the process of the TSLD-derivation step. We say that a TSLD-derivation is {\em successful} if we reach the empty query, further denoted by $\square$. The composition of the mgus $\theta_1, \dots, \theta_n$ used in each TSLD-derivation step is the {\em computed answer substitution} of the query. A TSLD-derivation that reaches $false$ is called a {\em failed derivation} and a TSLD-derivation that reaches $wrong$ is called an {\em erroneous derivation}.

In a TSLD-derivation, at each TSLD-derivation step we have several choices. We choose an atom from the query, a clause from the program, and a mgu. It is proven in \cite{Apt:1996:LPP:249573} that the choice of mgu does not affect the success or failure of an SLD-derivation, as long as the resulting mgu is idempotent. Since for TSLD-derivations the successes are the same as the ones in a corresponding SLD-derivation, then the result still holds for TSLD.

The {\em selection rule}, i.e, how we choose the selected atom in the considered query, does not influence the success of a TSLD-derivation either \cite{Apt:1996:LPP:249573}, however if you stopped as soon as unification returns false, it could prevent us from detecting a type error, in a later atom. Let us show this in the following example.

\begin{example}\label{exfalsewrong}
Consider the logic program $P$ consisting of only one fact, $p(X,X)$, and the selection rule that chooses the leftmost atom at each step.

Then, if we stopped when reaching $false$, the query $Q = p(1,2),p(1,a)$ would have the TSLD-derivation $Q \implies false$, since typed unification between $p(X,X)$ and $p(1,2)$ outputs $false$. However, the query $Q\prime = p(1,a),p(1,2)$ has the TSLD-derivation $Q \implies wrong$, since typed unification between $p(X,X)$ and $p(1,a)$ outputs $wrong$.
\end{example}

In fact, as the comma stands for conjunction, and since $wrong \wedge false = false \wedge wrong = wrong$, we have to continue even if typed unification outputs $false$ in a step, and check if we ever reach the value $wrong$. In general, for any selection rule $S$ we can construct a query $Q$ such that it is necessary to continue when typed unification outputs $false$ for some atom in $Q$.
Therefore, when we reach the value $false$ in a TSLD-derivation step, we continue applying steps until either we obtain a value $wrong$ from typed unification or we have no more atoms to select. In this last case, we can safely say that we reached $false$. This guarantees independence of the selection rule. For the following example we use the selection rule that always chooses the leftmost atom in a query, which is the selection rule of Prolog.

\begin{example}
Let us continue example \ref{exfalsewrong}. The TSLD-derivation for $Q$ is $Q \implies false,p(1,a) \implies wrong$. Let $Q\prime\prime = p(1,2),p(1,1)$. Then the TSLD-derivation is $Q\prime\prime \implies false,p(1,1) \implies false$. 
\end{example}
Note that when we get to $false$ for a typed unification in a TSLD-derivation, we can only output $false$ or $wrong$, so either way it is not a successful derivation.

The selected clause from the program is another choice point we have at each TSLD-derivation step. We will discuss the impact of this choice in the next section.

\subsection{TSLD-tree}

When we want to find a successful TSLD-derivation for a query, we need to consider the entire search space, which consists of all possible derivations, choosing all possible clauses for a selected atom. We are considering a fixed selection rule here, so the only thing that changes between derivations is the selected clause.
We say that a clause $H \leftarrow \bar{B}$ is {\em applicable} to an atom $A$ if $H$ and $A$ have the same predicate symbol with the same arity.
\begin{definition}[TSLD-tree]
Given a program $P$ and a query $Q$, a {\em TSLD-tree} for $P \cup \{Q\}$ is a tree where branches are TSLD-derivations of $P \cup \{Q\}$ and  every node $Q$ has a child for each clause from $P$ applicable to the selected atom of $Q$. 
\end{definition}
Prolog uses the leftmost selection rule, where one always selects the leftmost atom in the query, and since the selection rule does not change the success of a TSLD-derivation we will use this selection rule in the rest of the paper.

\begin{definition}[TSLD-tree classification]
\begin{itemize}
    \item If a TSLD-tree contains the empty query, we call it {\em successful}.
    \item If a TSLD-tree is finite and all its branches are erroneous TSLD-derivations, we call it {\em finitely erroneous}.
    \item If a TSLD-tree is finite and it is not successful nor finitely erroneous, we say it is {\em finitely failed}.
\end{itemize}
\end{definition}

\begin{example}
Let program $P$ be:
\begin{verbatim}
p(1).
p(2).
q(1).
q(a).
r(X) :- p(X),q(X).
\end{verbatim}
and let query $Q$ be $r(1)$. The TSLD-tree for $Q$ and $P$ is the following successful TSLD-tree:
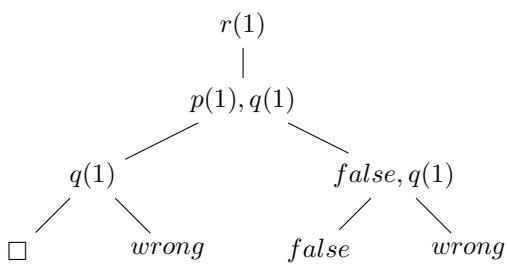
\begin{center}
\begin{tikzpicture}
[level 1/.style = {sibling distance = 4cm},
level 2/.style = {sibling distance = 4cm},
level 3/.style = {sibling distance = 2cm}]
    \node {$r(1)$} [level distance = 1cm]
        child {node {$p(1),q(1)$}
            child {node {$q(1)$} 
                child {node {$\square$} 
                }
                child {node {$wrong$} 
                }
            }
            child {node {$false,q(1)$} 
                child {node {$false$} 
                }
                child {node {$wrong$}
                }
            }
        };  
\end{tikzpicture}
\end{center}
\end{example}
We will present two auxiliary definitions which are needed to clearly define the notion of a type error in a program.

\begin{definition}[Generic Query]
Let $Q$ be a query and $P$ a program. We say that $Q$ is a {\em generic query} of $P$ iff $Q$ is only composed of an atom of the form $p(X_1,\dots,X_n)$ for each predicate symbol $p$ that occurs in the head of at least one clause in $P$, where $X_1,\dots,X_n$ are variables that occur only once in the query, and there are no other atoms in $Q$.
\end{definition}

\begin{example}
Let $P$ be the program defined as follows:
\begin{verbatim}
p(X,X).
q(X) :- p(1,a).
\end{verbatim}
Then, given the generic query $p(X_1,X_2),q(X_3)$, we have the following TSLD-derivation: $p(X_1,X_2),q(X_3) \implies q(X_3) \implies p(1,a) \implies wrong$
\end{example}

\begin{definition}[Blamed Clause]
Given a program $P$ and a query $Q$, a clause $c$ is a {\em blamed clause} of the TSLD-tree for $P \cup \{Q\}$ if all derivations where $c$ is a input clause are erroneous.
\end{definition}
The blamed clause is a clause in the program which causes a type error. A similar notion was first defined for functional programming languages with the blame calculus \cite{DBLP:conf/esop/WadlerF09}.

\begin{example}
Let $P$ be the following program, with clauses $c_1$, $c_2$, and $c_3$, respectively:
\begin{verbatim}
p(1).
q(a).
q(X) :- p(a).
\end{verbatim}
Then for the query $p(2),q(b)$, we have the following TSLD-tree:

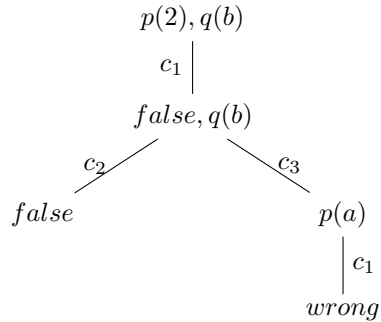
\begin{center}
\begin{tikzpicture}
[level 1/.style = {sibling distance = 4cm},
level 2/.style = {sibling distance = 4cm},
level 3/.style = {sibling distance = 2cm}]
    \node {$p(2),q(b)$} [level distance = 1.3cm]
        child {node {$false,q(b)$}
            child {node {$false$} edge from parent node [left] {$c_2$}
            }
            child {node {$p(a)$} 
                child {node {$wrong$} edge from parent node [right] {$c_1$}
                }
            edge from parent node [right] {$c_3$}
            }
        edge from parent node [left] {$c_1$}
        };
\end{tikzpicture}
\end{center}

In this case, $c_3$ is a blamed clause, since every derivation that uses it eventually reaches $wrong$. Note that $c_1$ is not a blamed clause, because the leftmost branch of the TSLD-tree uses $c_1$ but is $false$.
\end{example}

\begin{definition}[Type Error in the Program]
Suppose we have a program $P$ and a generic query $Q$. Then $P$ has a {\em type error} if there is a blamed clause in the TSLD-tree for $P \cup \{Q\}$.
\end{definition}

Note that if a program does not have a type error, then there is no \emph{blamed clause} in the TSLD-tree. Also note that the order of the atoms in the generic query may change the place where we use the blamed clause, but it will not change the fact that there exists one or not.

\begin{example}
Assume the same program from Example 8. Let $Q_1 = p(X1),q(X2)$ and $Q_2 = q(X1),p(X2)$ be generic queries. The TSLD-trees for $P \cup \{Q_1\}$ and $P \cup \{Q_2\}$ are:
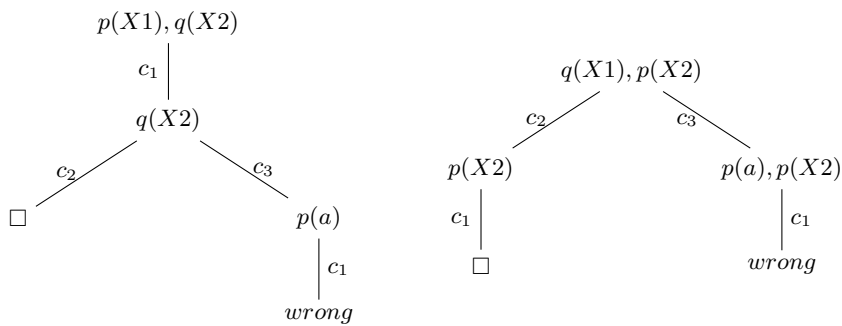
\begin{table}[!htb]
    \begin{minipage}{.50\linewidth}
        \begin{center}
        \begin{tikzpicture}
        [level 1/.style = {sibling distance = 4cm},
        level 2/.style = {sibling distance = 4cm},
        level 3/.style = {sibling distance = 2cm}]
            \node {$p(X1),q(X2)$} [level distance = 1.3cm]
                child {node {$q(X2)$}
                    child {node {$\square$} edge from parent node [left] {$c_2$}
                    }
                    child {node {$p(a)$} 
                        child {node {$wrong$} edge from parent node [right] {$c_1$}
                        }
                    edge from parent node [right] {$c_3$}
                    }
                edge from parent node [left] {$c_1$}
                };
        \end{tikzpicture}
        \end{center}
    \end{minipage}
    \begin{minipage}{.50\linewidth}
        \begin{center}
        \begin{tikzpicture}
        [level 1/.style = {sibling distance = 4cm},
        level 2/.style = {sibling distance = 4cm},
        level 3/.style = {sibling distance = 2cm}]
            \node {$q(X1),p(X2)$} [level distance = 1.3cm]
                child {node {$p(X2)$}
                    child {node {$\square$} edge from parent node [left] {$c_1$}
                    }
                edge from parent node [left] {$c_2$}
                }
                child {node {$p(a),p(X2)$}
                    child {node {$wrong$} edge from parent node [right] {$c_1$}
                    }
                edge from parent node [left] {$c_3$}
                };
        \end{tikzpicture}
        \end{center}
    \end{minipage}
    \\
\label{ex9}
\end{table}

\end{example}

Intuitively, having a type error in the program means that somewhere in the program we will perform typed unification between two terms that do not have the same type.

Consider a generic query $Q = A,\bar{A}$, where $A = p(X_1,\dots,X_n)$. For some derivation, after one step, we will have $\theta(\bar{B},\bar{A})$, where $H\leftarrow \bar{B}$ is a clause in $P$ and $\theta$ is a unifier of $A$ and $H$. Since $\theta$, or any other idempotent mgu of $A$ and $H$, is a renaming of $\{X_1 \mapsto t_1,\dots,X_n \mapsto t_n\}$, where $H = p(t_1,\dots,t_n)$, and since the variables $X_1,\dots,X_n$ do not occur in $\bar{B}$ because the clause is variable disjoint from the query by definition, nor do they occur in $\bar{A}$, by the definition of the generic query, then $\theta(\bar{B},\bar{A}) = \bar{B},\bar{A}$. The argument holds for any other atom in $\bar{A}$. Thus whenever the selected atom belongs to the original query $Q$ the result is never $wrong$ and the mgu can be ignored.

After selecting the blamed clause $c$, every TSLD-derivation is such that $Q_0 \implies \dots \implies Q_n \implies \emph{wrong}$. Thus, at step $Q_n$, the selected atom comes from the program and every mgu applied up to this point is from substitutions arising from the program itself and not the query. Therefore, the type error was in the program.

\begin{definition}[Type Error in the Query]
Let $P$ be a program and $Q$ be a query. If there is no type error in $P$ and the TSLD-tree is finitely erroneous, then we say that there is a {\em type error in the query} $Q$ with respect to $P$.
\end{definition}

If there is no type error in the program $P$ but the TSLD-tree is finitely erroneous, then that error must have occurred in a unification between terms from the query and the program. We then say that the type error is in the query.

\begin{example}
Now suppose we have the following program:
\begin{verbatim}
p(1).
q(a).
q(X) :- p(X).
\end{verbatim}
Let us name the clauses $c_1$, $c_2$, and $c_3$, respectively. The following trees are the TSLD-tree for a generic query, and the TSLD-tree for the query $Q\prime = q(1.1)$.

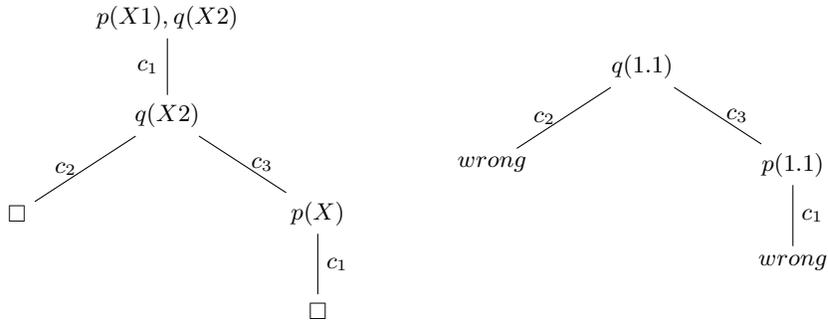
\begin{table}[!htb]
    \begin{minipage}{.50\linewidth}
        \begin{center}
        \begin{tikzpicture}
        [level 1/.style = {sibling distance = 4cm},
        level 2/.style = {sibling distance = 4cm},
        level 3/.style = {sibling distance = 2cm}]
            \node {$p(X1),q(X2)$} [level distance = 1.3cm]
                child {node {$q(X2)$}
                    child {node {$\square$} edge from parent node [left] {$c_2$}
                    }
                    child {node {$p(X)$} 
                        child {node {$\square$} edge from parent node [right] {$c_1$}
                        }
                    edge from parent node [right] {$c_3$}
                    }
                edge from parent node [left] {$c_1$}
                };
        \end{tikzpicture}
        \end{center}
    \end{minipage}
    \begin{minipage}{.50\linewidth}
        \begin{center}
        \begin{tikzpicture}
        [level 1/.style = {sibling distance = 4cm},
        level 2/.style = {sibling distance = 4cm},
        level 3/.style = {sibling distance = 2cm}]
            \node {$q(1.1)$} [level distance = 1.3cm]
                child {node {$wrong$}
                edge from parent node [left] {$c_2$}
                }
                child {node {$p(1.1)$}
                    child {node {$wrong$}
                    edge from parent node [right] {$c_1$}
                    }
                edge from parent node [right] {$c_3$}
                };
        \end{tikzpicture}
        \end{center}
    \end{minipage}
    \\
\label{ex10}
\end{table}

On the left-hand side, we see that for a generic query $Q$, the TSLD-tree for $P \cup \{Q\}$ has no blamed query, which means that there is no type error in the program. On the right-hand side the tree is finitely erroneous, therefore there is a type error in query $Q\prime$.
\end{example}

\section{Declarative Semantics}

The declarative semantics of logic programming is, in opposition to the operational one, a definition of what the programs compute. The fact that logic programming can be interpreted this way supports the fact that logic programming is declarative \cite{Apt:1996:LPP:249573}. In this section, we will introduce the concept of interpretations, which takes us from the syntactic programs we saw and used so far into the semantic universe, giving them meaning. With this interpretation we will redefine a declarative semantics for logic programming first defined in \cite{BarbosaFloridoCosta19} and prove a connection between both the operational and the declarative semantics.

\subsection{Domains}

Let $U$ be a non-empty set of semantic values, which we will call the universe. We assume that the universe is divided into domains such that each ground type is mapped to a non-empty domain. Thus, $U$ is divided into domains as follows: $U = Int + Float + Atom + String + A_1 + \dots + A_n + F + Bool + W$, where $Int$ is the domain of integer numbers, $Float$ is the domain of floating point numbers, $Atom$ is the domain of non-numeric constants, $String$ is the domain of strings, $A_i$ are domains for trees, where each domain has trees whose root is the same functor symbol and its n-children belong to n domains and $F$ is the domain of functions. Moreover, we define $Bool$ as the domain containing $true$ and $false$, and $W$ as the domain with the single value $wrong$, corresponding to a run-time error. We will call $Int$, $Float$, $Atom$, and $String$ the base domains, and $A_1, \dots, A_n$ the tree domains.
In particular, we can see that constants are separated into several predefined base domains, one for each base type, while complex terms, i.e. trees, are separated into domains depending on the principal function symbol (root) and the n-tuple inside the parenthesis (n-children).

\subsection{Interpretations}

 Every constant of type $T$ is associated with a semantic value in one of the base domains, $Int$, $Float$, $Atom$, or $String$, corresponding to $T$. Every function symbol $f$ of arity $n$ in our language is associated with a mapping $f_U$ from any n-tuple of base or tree domains $\delta_1 \times \dots \times \delta_n$ to the domain $F(\delta_1,\dots,\delta_n)$, which is the domain of trees whose root is $f$ and the n-children are in the domains $\delta_i$. %This mapping is polymorphic.

To define the semantic value for terms, we will first have to define states. States, $\Sigma$, are mappings from variables into values of the universe. We also define a function $domain$ that when applied to a semantic value returns the domain it belongs to. The semantic value of a term is defined as follows:\\

\noindent
$|[ X|]_{\Sigma} = \Sigma(X)$\\
$|[ c|]_{\Sigma} = c_U$\\
$|[ f(t_1,\dots,t_n)|]_{\Sigma} = f_U(|[ t_1|]_{\Sigma},\dots,|[ t_n|]_{\Sigma})$\\

An {\em interpretation} associates every predicate symbol $p$ with a function $p_U$ in $F$, such that the output of the function $p_U$ is the domain $Bool$ and the input is a union of tuples of domains. For each tuple that is in its domain, the function $p_U$ either returns $true$ or $false$. We will use $|[~|]_{I,\Sigma}$ to denote the semantics of an {\em expression} $E$, which can be an atom, a query, or a clause, in an interpretation $I$, and define it as follows:\\

% \noindent$|[ t_1 = t_2|]_{I,\Sigma}$~=~{\bf if} ($|[ t_1|]_{\Sigma} = |[ t_2|]_{\Sigma}$)\\
% \hspace*{2.5cm} \textbf{then} $true$\\
% \hspace*{2.5cm} \textbf{else if} $(domain(|[ t_1|]_{\Sigma}) = domain(|[ t_2|]_{\Sigma})$) \\
% \hspace*{4cm} \textbf{then} $false$\\
% \hspace*{4cm} \textbf{else} $wrong$\\
\noindent
$|[ p(t_1,\dots,t_n)|]_{I,\Sigma}$~=~{\bf if} $(domain(|[ t_1|]_{\Sigma}),\dots,domain(|[ t_n|]_{\Sigma})) \subseteq domain(I(p))$ \\
\hspace*{3.5cm} \textbf{then} $I(p)(|[ t_1|]_{\Sigma},\dots,|[ t_n|]_{\Sigma})$ \\
\hspace*{3.5cm} \textbf{else} $wrong$\\
$|[ A_1,\dots,A_n|]_{I,\Sigma}~=~|[ A_1|]_{I,\Sigma} \wedge \dots \wedge |[ A_n|]_{I,\Sigma}$\\
$|[ q(t_1,\dots,t_n) :- \bar{B}|]_{I,\Sigma}~=~(|[ \bar{B}|]_{I,\Sigma} --> (|[ q(t_1,\dots,t_n)|]_{I,\Sigma}))$\\ \\
Note that if the clause is of the form $H \leftarrow$, then its semantics is equivalent to that of $H$.

\subsection{Models}

The term language and their semantic values are fixed, thus each interpretation $I$ is determined by the interpretation of the predicate symbols. Interpretations differ from each other only in the functions $p_U$ they associate to each predicate $p$ defined in $P$.
We now define a {\em context} as a set $\Delta$ of pairs of the form $X : D$, where $X$ is a variable that occurs only once in the set, and $D$ is a domain. We say that $\Sigma$ complies with $\Delta$ if every binding $X : v$ in $\Sigma$ is such that $(X : D) \in \Delta$ and $v \in D$.
An interpretation $I$ is a {\em model} of $E$ in the context $\Delta$ iff for every state $\Sigma$ that complies with $\Delta$, $|[E|]_{I,\Sigma}=true$. We will denote this as $\Delta \models |[E|]_I$. Given a program $P$, we say that an interpretation $I$ is a model of $P$ in context $\Delta$ if $I$ is a model of every clause in $P$ in context $\Delta$. Here we assume, without loss of generality, that all clauses are variable disjoint with each other.
If two expressions $E_1$ and $E_2$ are such that every model of $E_1$ in a context $\Delta$ is also a model of $E_2$ in the context $\Delta$, then we say that $E_2$ is a semantic consequence of $E_1$ and represent this by $E_1 \models E_2$.
Suppose two interpretations $I_1$ and $I_2$ are models of program $P$ in some context $\Delta$. Suppose, in particular, that for some predicate $p$ of $P$ the associated function is $p_{U1}$ for $I_1$ and $p_{U2}$ for $I_2$. Let us call $T_i$ the set of tuples of terms for which $p_{Ui}$ outputs $true$, and $F_i$ to the set of tuples of terms for which $p_{Ui}$ outputs $false$. We say that $I_1$ is smaller than $I_2$ if $T_1\subseteq T_2$ and, if $T_1 = T_2$, then $F_1 \subseteq F2$.
We say that a model $I$ of $P$ in context $\Delta$ is {\em minimal} if for every other model $I\prime$ of $P$ in context $\Delta$, $I$ is smaller than $I\prime$.

\begin{example}
Consider the program $P$ defined below:
\begin{verbatim}
father(john,mary).
father(phil,john).
grandfather(X,Y) :- father(X,Z), father(Z,Y).
\end{verbatim}
Suppose that interpretation $I_1$ is such that $p_{U_1}$ is associated with $grandfather$ and $p_{U1} :: Atom \times Atom \to Bool$. Suppose that $p_{U2}$ is associated to $grandfather$ in $I_2$, with the same domain. Suppose, also, that $p_{U3}$, associated in $I_3$ with $grandfather$, is such that $p_{U3} :: Atom \times Atom \cup Int \times Int \to Bool$. Let the sets $T_1 = \{(phil,mary)\}$, $T_2 = \{(phil,mary), (john,caroline)\}$, and $T_3 = \{(phil,mary)\}$ be the sets of accepted tuples for $p_{U1}$, $p_{U2}$, and $p_{U3}$, respectively.

Thus, if these interpretations associate the same function $q_U :: Atom \times Atom \to Bool$ to $father$, and $T = \{(john,mary), (phil,john)\}$ the set of accepted tuples for $q_U$, then all $I_i$ are models of $P$ in context $\Delta = \{X: Atom, Y : Atom, Z : Atom\}$. In fact, all states $\Sigma$ that comply with $\Delta$ are such that $|[grandfather(X,Y) :- father(X,Z),father(Z,Y)|]_{I_i,\Sigma}$ is true, for all $i = 1,2,3$.

But note that $T_1 \subseteq T_2$, and $T_1 = T_3$, but $F_1 \subseteq F_3$. In fact, any smaller domain or set $T_k$ would not model $P$. Therefore $I_1$ is the minimal model of $P$.
\end{example}

\subsection{Type Errors}

To calculate the set of accepted tuples for a given interpretation we will use the immediate consequence operator $T_P$. The $T_P$ operator is traditionally used in the logic programming literature to iteratively calculate the minimal model of a logic program as presented in \cite{DBLP:journals/jacm/EmdenK76,Apt:1996:LPP:249573,Lloyd:1984:FLP:2214}.

Since for us interpretations for predicates are typed, $T_P\uparrow \omega(\emptyset)$ does not generate an interpretation by itself. Instead it generates a set of atoms $S$. Then we say that any interpretation $I$ derived from $S$ is such that for all predicates $p$ occurring in $S$, $I(p) :: (D_{(1,1)} \times \dots \times D_{(1,n)}) \cup \dots \cup (D_{(k,1)},\dots,D_{(k,n)}) \to Bool$, where for all $i$ there is at least one atom $p(v_1,\dots,v_n) \in S$ such that $v_1\in D_{(i,1)},\dots,v_n\in D_{(i,n)}$. Note that these interpretations may not be models of $P$ using our new definition of a model. We are now able to define the notion of {\em ill-typed program}.

\begin{definition}[Ill-typed Program]
Let $P$ be a program. If no interpretation derived from $T_P \uparrow \omega(\emptyset)$ is a model of $P$, we say that $P$ is an ill-typed program.
\end{definition}

\begin{example}
Let $P$ be the program defined as:
\begin{verbatim}
p(1).
p(a).
q(X) :- p(1.1).
\end{verbatim}
Then $S = T_P\uparrow \omega(\emptyset) = \{p(1),p(a)\}$. So any interpretation $I$ derived from $S$ is such that $p_I :: Int \cup Atom \to Bool$. Therefore for any context $\Delta$, for every $\Sigma$ that complies with $\Delta$, $|[q(X) :- p(1.1)|]_{I,\Sigma} = wrong$. Therefore no such $I$ is a model of $P$.
\end{example}

The reason why $T_P\uparrow \omega(\emptyset)$ is always a minimal model of $P$ in the untyped semantics, comes from the fact that whenever a body of a clause is $false$ for all states, then the clause is trivially $true$ for all states. However in our semantics, since we are separating these cases into $false$ and $wrong$, the $wrong$ ones do not trivially make the formula $true$, making it $wrong$ instead. These are the ill-typed cases.

\begin{lemma}[Blamed Clause Type Error]\label{bcc}
Suppose there is a type error in the program with blamed clause $c = H \leftarrow A_1,\dots,A_m$. Then, $\exists A_i = p(t_1,\dots,t_n)$ such that $\forall p(s_1,\dots,s_n) \in T_P \uparrow \omega (\emptyset). \forall \Sigma .\exists j. domain(|[t_j|]_{\Sigma}) = domain(|[s_j|]_{\Sigma})$.
\end{lemma}

Effectively what this means is that if there is a type error in the program, then the blamed clause is such that it will not be used to calculate the $T_P \uparrow \omega (\emptyset)$, since at least one of the atoms in its body will never be able to be used in an application of $T_P$. We can also have a ill-typed query, and we define it as follows.

\begin{definition}[Ill-typed Query]
Let $P$ be a program. If any interpretation $I$ derived from $T_P\uparrow \omega(\emptyset)$, such that $I$ models $P$ in some context $\Delta$, is such that $I$ is not a model of $Q$ in the context $\Delta$, then we say that $Q$ is an ill-typed query with respect to $P$.
\end{definition}

\subsection{Soundness of TSLD-resolution}

In this section we will prove that TSLD-resolution is sound, i.e. if there is a successful derivation of a query $Q$ in program $P$ with a correct answer substitution $\theta$, then every model of $P$ is also a model of $\theta(Q)$; if there is a type error in the program, then the program is ill-typed; and if there is a type error in the query, the query is ill-typed with respect to the program.
To prove this we will introduce the following auxiliary concept.

\begin{definition}[Resultant]
Suppose we have a TSLD-derivation step $Q_1 \implies \theta(Q_2)$. Then we define the resultant associated to this step as $\theta(Q_1) \leftarrow Q_2$.
\end{definition}

\begin{lemma}[Soundness of resultants]\label{sor}
Let $Q_1 \implies \theta(Q_2)$ be a TSLD-derivation step using input clause $c$ and $r$ be the resultant associated with it. Then:
\begin{enumerate}
\item $c \models r$; 
\item for any TSLD-derivation of $P \cup \{Q\}$ with resultants $r_1,\dots,r_n$, $P\models r_i$ (for all $i \geq 0$).
\end{enumerate}
\end{lemma}

Proof of this lemma for the SLD-resolution is in \cite{Apt:1996:LPP:249573}. Since for unifiable terms the typed unification algorithm behaves like first-order unification, the proof still holds.

\begin{theorem}[Soundness of TSLD-resolution]
\label{sound}
Let $P$ be a program and $Q$ a query. Then:
\begin{enumerate}
\item Suppose that there exists a successful derivation of $P \cup \{Q\}$, with the correct answer substitution $\theta$. Then $P \models \theta(Q)$.
\item Suppose there is a type error in the program. Then $P$ is ill-typed.
\item Suppose there is a type error in the query. Then $Q$ is an ill-typed query with respect to $P$.
\end{enumerate}
\end{theorem}
A short note about completeness. As for untyped SLD-resolution, completeness is related to the search for answers in a TSLD-tree. If we use Prolog sequential, top-down, depth-first search with backtracking, then it may result in incompleteness for same cases where the TSLD-tree is infinite, because the exploration of an infinite computation may defer indefinitely the exploration of some alternative computation capable of yielding a correct answer.

\section{Conclusions and Future Work}

We presented an operational semantics for logic programming, here called TSLD-resolution, which is sensitive to run-time type errors. In this setting type errors are represented by a new value, here called {\em wrong}, which is added to the usual {\em fail} and {\em success} results of evaluation of a query for a given logic program.  We have then adapted a previously defined declarative semantics for typed logic programs using a three-valued logic and proved that TSLD-resolution is sound with respect to this semantics.
All these new concepts, TSLD, typed unification and the new declarative semantics, revisit and partially extend well-known concepts form the theory of logic programming.

{\bf Specially interpreted functors:}
in this paper functors are uninterpreted, such as in Prolog, in the sense that they are just symbols used to build new trees. An obvious extension of this work is to extend the system to dynamically detect type errors relating to the semantic interpretation of some specific functors, for instance the list constructor. For this, we would have for the list constructor not the Herbrand-based interpretation $[~|~] :: \forall A,B.A \times B \rightarrow [A | B]$, but the following interpretation $[~|~] :: \forall A. A \times list(A) \rightarrow list(A)$. Moreover, we would have the empty list $[~]$ with type $\forall D. list(D)$. This would necessarily change the typed unification algorithm by introducing a new kind of constraints. As an example, consider the unification $[1|2] = [1|2]$, where the second argument in both terms is not a list: considering a specially interpreted list constructor the result should be \emph{wrong}, although the traditional untyped result is {\em true}. 
The same issues appear for arithmetic expressions. Arithmetic interpretations of $+, -, \times$, and $/$ can be introduced in the typed unification algorithm, so that in this context, unifications such as  $a + b = a + b$ would now return \emph{wrong} instead of {\em true}.
These extensions are left for future work.

{\bf Using TSLD in practice:}
we plan to integrate TSLD into the YAP Prolog System \cite{YAP}. This will enable further applications of the method to large scale Prolog programs. One important point of this integration it to add explicit type declarations to all the implicitly typed Prolog built-in predicates. Finally, one could also wonder how to apply TSLD to the dynamic typing of constraint logic programming modules, adding new types mapped to constraint domains.

\bibliographystyle{alpha}
\bibliography{bibliography}

\newpage
\appendix

\section{Proofs for Theorems}
\setcounter{theorem}{0}

\begin{theorem}[Equivalence to term unification for success]
Let $t_1$ and $t_2$ be two terms. If we apply the Martelli-Montanari algorithm (MM algorithm) to $t_1$ and $t_2$ and it returns a set of solved equalities $S$, then the typed unification algorithm applied to the same two terms is also successful and returns the same set of equalities.
\end{theorem}

\begin{proof}
The proof follows from induction on $S$.
Base cases:
\begin{itemize}
    \item Suppose we have an equality in $S$ of the form $f(t_1,\dots,t_n) = g(s_1,\dots,s_m)$. Then the MM algorithm halts with failure, and our algorithm halts with $wrong$.

    \item Suppose we have an equality in $S$ of the form $c = d$. Then the MM algorithm halts with failure, and our algorithm either halts with $wrong$ or changes $F$ to $false$, depending on the types of $c$ and $d$. In either case, it will not be successful.
    
    \item Suppose we have an equality in $S$ of the form $c = f(t_1,\dots,t_n)$ or $f(t_1,\dots,t_n) = c$. Then the MM algorithm halts with failure, and our algorithm halts with $wrong$.
    
    \item Suppose we have an equality in $S$ of the form $X = t$, where $X$ occurs in $t$. Then, the MM algorithm halts with failure and our algorithm changes $F$ to $false$, so it is never successful.
\end{itemize}

Inductive cases:
\begin{itemize}
    \item Suppose we have an equality in $S$ of the form $f(t_1,\dots,t_n) = f(s_1,\dots,s_n)$. Both algorithms generate the same new equalities and replace the selected one with those in $S$. Then, by the induction hypothesis, if the MM algorithm succeeds and outputs a set of solved equalities $S\prime$, so does our algorithm.
    
    \item Suppose we have an equality in $S$ of the form $c = c$. Both algorithms delete this equality from $S$. Then, by the induction hypothesis, if the MM algorithm succeeds and outputs a set of solved equalities $S\prime$, so does our algorithm.
    
    \item Suppose we have an equality in $S$ of the form $X = X$. Both algorithms delete this equality from $S$. Then, by the induction hypothesis, if the MM algorithm succeeds and outputs a set of solved equalities $S\prime$, so does our algorithm.
    
    \item Suppose we have an equality in $S$ of the form $t = X$, where $t$ is not a variable and $X$ is a variable. Both algorithms replace this equality from $S$ with the same new one. Then, by the induction hypothesis, if the MM algorithm succeeds and outputs a set of solved equalities $S\prime$, so does our algorithm.
    
    \item Suppose we have an equality in $S$ of the form $X = t$, where $X$ does not occur in $t$ and $X$ occurs somewhere else in $S$. Both algorithms apply the same substitution to $S\setminus\{X = t\}$, therefore resulting in the same set of equalities. Then, by the induction hypothesis, if the MM algorithm succeeds and outputs a set of solved equalities $S\prime$, so does our algorithm.
\end{itemize}

In any other case, none of the algorithms apply. \qed
\end{proof}

\begin{theorem}[Ill-typed unification]
If the output of the typed unification algorithm is \emph{wrong}, then there is no substitution $\theta$ such that $\theta(t_1)$ and $\theta(t_2)$ have the same type.
\end{theorem}

\begin{proof}
Suppose that there is a substitution $\theta$ such that $\theta(t_1) = \theta(t_2)$. Then, we now that the MM algorithm would output an mgu of $t_1$ and $t_2$. Therefore so would our algorithm, from theorem \ref{equi}. But the output of our algorithm is \emph{wrong}, so there cannot be such a $\theta$.
Suppose that there is a substitution $\theta$ such that $\theta(t_1)$ and $\theta(t_2)$ have the same type. Then, we will prove that our algorithm would output $false$ by induction on $t_1$ and $t_2$.
Base cases:
\begin{itemize}
    \item Suppose we have $c = d$ and $c$ has the same type as $d$. Then the algorithm outputs $false$, not $wrong$.
    
    \item Suppose we have $c = d$ and $c$ has a different type from $d$. Then such $\theta$ cannot exist.

    \item Suppose we have $X = t$, where $X$ occurs in $t$. Then the algorithm outputs $false$, not $wrong$.

    \item Suppose we have $f(t_1,\dots,t_n) = g(s_1,\dots,s_m)$. Then such $\theta$ cannot exist.
\end{itemize}
Inductive step:
\begin{itemize}
    \item Suppose we have $f(t_1,\dots,t_n) = f(s_1,\dots,s_n)$. Then, by the induction hypothesis, if the algorithm outputs $wrong$ for the input $(\{t_1=s_1,\dots,t_n=s_n\},true)$, we know that no $\theta$ such that $\forall i. \theta(t_i)$ and $\theta(s_i)$ have the same type, exists. Therefore, no $\theta$ such that $\theta(f(t_1,\dots,t_n))$ and $\theta(f(s_1,\dots,s_n))$ have the same type exists either, by the properties of substitution.
    
    \item Suppose we have $c = c$ or $X = X$, then the algorithm outputs $true$, not $wrong$.
    
    \item Suppose we have $X = t$, then either the algorithm ends with $F = true$ or $F = false$ if $X$ occurs in $t$, but never $wrong$.
    
    \item Suppose we have $t = X$, where $t$ is not a variable. Then, the algorithm either halts with success in one step or fails. Either way it does not output $wrong$.
\end{itemize}
So, we prove that whenever the typed unification algorithm halts with $wrong$ for some pair of terms $t_1$ and $t_2$, then there is no substitution $\theta$ such that $\theta(t_1)$ and $\theta(t_2)$ have the same type. \qed
\end{proof}

\begin{lemma}[Blamed Clause Type Error]
Suppose there is a type error in the program with blamed clause $c = H \leftarrow A_1,\dots,A_m$. Then, $\exists A_i = p(t_1,\dots,t_n)$ such that $\forall p(s_1,\dots,s_n) \in T_P \uparrow \omega (\emptyset). \forall \Sigma \exists j. domain(|[t_j|]_{\Sigma}) = domain(|[s_j|]_{\Sigma})$.
\end{lemma}
\begin{proof}
We will prove this by contradiction. Suppose that for all $A_i$ there is some $p(s_1,\dots,s_n) \in T_P \uparrow \omega (\emptyset)$ such that for some $\Sigma$, $\forall i \in [1,\dots,n] domain(|[t_i|]_{\Sigma}) = domain(|[s_i|]_{\Sigma})$. Then, there would be a derivation of the form $A_1,\dots,A_m, \bar{B} \implies \dots \implies \bar{B}$ or $A_1,\dots,A_m, \bar{B} \implies \dots \implies false,\bar{B}$ in the TSLD-tree for $P \cup \{Q\}$, where $Q$ is a generic query, since the output for the unification between $A_i$ and $p(s_1,\dots,s_n)$ would not return $wrong$. But then $c$ would not be a blamed clause. Therefore we proved the lemma.
\end{proof}

\begin{theorem}[Soundness of TSLD-resolution]
Let $P$ be a program and $Q$ a query. Then:
\begin{enumerate}
\item Suppose that there exists a successful derivation of $P \cup \{Q\}$, with the correct answer substitution $\theta$. Then $P \models \theta(Q)$.
\item Suppose there is a type error in the program. Then $P$ is ill-typed.
\item Suppose there is a type error in the query. Then $Q$ is an ill-typed query with respect to $P$.
\end{enumerate}
\end{theorem}

\begin{proof}
(1) Let $\theta_1,\dots,\theta_n$ be the mgus obtained in a successful derivation. Therefore, $\theta = \theta_n \circ \dots \circ \theta_1$. The proof follows directly from lemma \ref{sor} applied to $P \models \theta_n \circ \dots \circ \theta_1(Q) \leftarrow \square$.

(2) If there is a type error in the program, then there is a blamed clause $c$ in the TSLD-tree for $P \cup \{GQ\}$, where $GQ$ is a generic query. Then by lemma \ref{bcc} we know that there is at least one $A_i = p(t_1,\dots,t_n)$, in the body of $c$, such that $\forall p(s_1,\dots,s_n) \in T_P \uparrow \omega (\emptyset). \exists i \in [1,\dots,n]. domain(|[t_i|]) = domain(|[s_i|])$.
Any interpretation $I$ derived from $T_P \uparrow \omega (\emptyset)$ is such that for any $\Sigma$ $|[A_i|]_{I,\Sigma} = wrong$. This implies that no such $I$ is a model of $c$ and, therefore, no such $I$ is a model of $P$, which means P is ill-typed.

(3) This is the case where there is a type error in the query $Q$. Now, consider program $P \cup \{p()\leftarrow Q\}$, where $p$ is a predicate that does not occur in $P$. Then note that every TSLD-derivation for a generic query with input clause $p() \leftarrow Q$ leads to $wrong$, so this is a blamed clause. Therefore, from (2), no interpretation $I$ derived from $T_P \uparrow \omega (\emptyset)$ models this clause. Since the set of atoms for $p$ in $T_P \uparrow \omega (\emptyset)$ is empty, we can give any interpretation to $p$ in any interpretation derived from $T_P \uparrow \omega (\emptyset)$. So we can choose an interpretation $I$ derived from $T_P \uparrow \omega (\emptyset)$ such that $I$ models $p()$ in some context $\Delta$. Therefore, as no interpretation derived from $T_P \uparrow (\emptyset)$ models $P \cup \{p()\leftarrow Q\}$ in any context $\Delta$ and we can build an interpretation which models predicate $p$ in some context, no interpretation can model $Q$. Thus, by the definition of ill-typed query, $Q$ is ill-typed with respect to $P$.

\qed
\end{proof}

\end{document}